\documentclass[12pt,draftclsnofoot,onecolumn,journal,letterpaper]{IEEEtran}

\usepackage{cite}
\usepackage{array}
\usepackage{bm}
\usepackage{booktabs}
\usepackage{yfonts}
\usepackage{caption}
\usepackage{subcaption}
\usepackage{amssymb}
\usepackage{multirow}
\usepackage{color}
\usepackage{graphicx}
\usepackage{algorithm}
\usepackage{algorithmic}
\usepackage{amsmath}
\usepackage{amsthm}
\usepackage{amssymb}
\usepackage{mathrsfs}
\theoremstyle{plain}
\newtheorem{theorem}{Theorem}
\newtheorem{lemma}[theorem]{Lemma}%[chapter]
\newtheorem{remark}[theorem]{Remark}%[chapter]

\definecolor{OXO-emph}{RGB}{153,0,0}
\title{Hierarchical Polar Coding for Achieving Secrecy over Fading Wiretap Channels without any Instantaneous CSI}
\author{\IEEEauthorblockN{Hongbo~Si, O.~Ozan~Koyluoglu, and Sriram~Vishwanath
\thanks{
The material in this paper was presented in part at 2015 IEEE International Symposium on Information Theory
(ISIT 2015), Hong Kong, June 2015.}
\thanks{H.~Si and S.~Vishwanath are with the Laboratory for Informatics, Networks, and Communications,
Wireless Networking and Communications Group,
The University of Texas at Austin, Austin, TX 78712. O.~O.~Koyluoglu is with the
Department of Electrical and Computer Engineering,
The University of Arizona, Tucson, AZ 85721.
Email: sihongbo@mail.utexas.edu, ozan@email.arizona.edu, sriram@austin.utexas.edu.}
}}

\begin{document}
\maketitle

%%%%%%%%%%%%%%%%%%%%%%%%%%%%%%%%%%%%%%%%%%%%%%%%%%%%%%%%%%%%%%%%%%%%%%%%%%%%%%
%%%%%%%%%%%%%%%%%%%%%%%%%%%%%%%%%%%%%%%%%%%%%%%%%%%%%%%%%%%%%%%%%%%%%%%%%%%%%%

\begin{abstract}

This paper presents a polar coding scheme to achieve secrecy in block fading binary symmetric wiretap channels without the knowledge of instantaneous channel state information (CSI) at the transmitter. For this model, a coding scheme that hierarchically utilizes polar codes is presented. In particular, on polarization of different binary symmetric channels over different fading blocks, each channel use is modeled as an appropriate binary erasure channel over fading blocks. Polar codes are constructed for both coding over channel uses for each fading block and coding over fading blocks for certain channel uses. In order to guarantee security, random bits are introduced at appropriate places to exhaust the observations of the eavesdropper. It is shown that this coding scheme, without instantaneous CSI at the transmitter, is secrecy capacity achieving for the simultaneous fading scenario. For the independent fading case, the capacity is achieved when the fading realizations for the eavesdropper channel is always degraded with respect to the receiver. For the remaining cases, the gap is analyzed by comparing lower and upper bounds. Remarkably, for the scenarios where the secrecy capacity is achieved, the results imply that instantaneous CSI does not increase the secrecy capacity.
\end{abstract}

\begin{keywords}
Fading Channel, Information Theoretic Security, Polar Coding, Wiretap Channel. 
\end{keywords}

%%%%%%%%%%%%%%%%%%%%%%%%%%%%%%%%%%%%%%%%%%%%%%%%%%%%%%%%%%%%%%%%%%%%%%%%%%%%%%
%%%%%%%%%%%%%%%%%%%%%%%%%%%%%%%%%%%%%%%%%%%%%%%%%%%%%%%%%%%%%%%%%%%%%%%%%%%%%%
\section{Introduction}

\subsection{Background}
Wiretap channels, introduced in the seminal paper of Wyner \cite{Wyner:wire75}, model the communication between a transmitter and a receiver in the presence of an eavesdropper that overhears the transmitted signals via the channel between transmitter and eavesdropper (e.g., by tapping the wire between the legitimate nodes). The task of transmitter is to hide information from the eavesdropper while communicating reliably to the receiver. Wyner studied this problem and characterized the capacity region for certain channel models including the case of degraded eavesdropper \cite{Wyner:wire75}. The achievability technique is the randomized version of the Shannon's random coding approach, where the randomization is utilized to confuse the eavesdropper, in order to achieve security. Since the publication of Wyner's work, several studies in the network information theory domain have utilized this random coding approach to characterize the corresponding secrecy capacities. Yet, the design of secrecy achieving coding schemes with practical constraints such as low complexity and availability of channel state information remains as an important direction in the physical layer security.

Recently, polar codes have been utilized for communication over degraded wiretap channels \cite{Mahdavifar:Wiretap11, Koyluoglu:Secure12, Andersson:Nested10, Hof:Secrecy10}. Polar codes are the first family of provably capacity achieving codes for symmetric binary-input discrete memoryless channels with low encoding and decoding complexity \cite{Arikan:Channel08}. These codes rely on the ``channel polarization'' technique, which reconstructs a set of equivalent channels such that each of them is either purely noisy or noiseless. Noting that the fraction of noiseless channels approaches the symmetric channel capacity, transmitting information symbols on the good instances and freezing the bad ones achieves the optimal rate. The schemes proposed in \cite{Mahdavifar:Wiretap11, Koyluoglu:Secure12, Andersson:Nested10, Hof:Secrecy10} are based on the behavior of the polarization of degraded channels, where the polarized channels for the degraded wiretap channels can be partitioned to one of the following sets: (i) good for both receiver and eavesdropper, (ii) good for receiver and bad for eavesdropper, and (iii) bad for both receiver and eavesdropper. The fraction of type (ii) channels approach to the secrecy capacity for the degraded (binary symmetric) wiretap channels, and the communication scheme utilizes this type of polarized indices to transmit information; whereas, type (i) channels are assigned to random bits to limit the eavesdropper's ability to obtain information about the messages. (Type (iii) channels are frozen, i.e., set to a constant value and shared to receiver.) This scheme allows for achieving the secrecy capacity, while inheriting the low complexity nature of polar codes. In other words, this technique mimics the Wyner's random coding approach with practical encoding/decoding schemes. The main hurdle for most practical applications though is to have the eavesdropper channel state information (CSI) at the transmitter, e.g., in order to differentiate between type (i) and (ii) channels in this coding scheme. Remarkably, an incorrect knowledge about the eavesdropper CSI would leak information, hence will not result in a meaningful security guarantee. In this work, we focus on relaxing the assumption on the instantaneous CSI knowledge, and develop polar coding schemes for fading wiretap channels, where only the statistics of CSI (of both receiver and eavesdropper) is known at the transmitter.

\subsection{Contributions}

The contributions of the paper can be summarized as follows:

\begin{itemize}
\item We first focus on a simultaneous fading model, where both main and eavesdropper channels experience the same fading states. A good exemplar for this scenario, perhaps, is physically degraded channels, where the eavesdropper observes the output of the main channel through an additional noisy channel (e.g., through wiretapping). In these models, when only the main channel experiences fading, the resulting system reduces to our simultaneous fading model. Here, we divide our analysis with respect to different orderings between the channel qualities seen by receiver and eavesdropper. Focusing on each case separately, we propose achievable secrecy rates based on hierarchical polar coding. This technique exploits multiple polar codes, which are utilized over channel uses for each block as well as over fading blocks for certain polarized channel indices.
\item We next focus on the optimality of the proposed scheme by comparing the resulting rate with an upper bound on secrecy capacity of the simultaneous fading model. The upper bound is obtained by allowing the transmitter to know instantaneous CSI, and remarkably, this bound is shown to be matching to the secrecy rates attained by the proposed achievable scheme, therefore characterizing the secrecy capacity of the system. It is remarkable that instantaneous CSI does not improve secrecy capacity for this model, and having only the statistical CSI suffices. (We note that our proposed coding scheme based on polar codes is instrumental in obtaining this result, as, to the best of our knowledge, there are no random coding strategies achieving a similar performance.)
\item Thirdly, we focus on a general model, where both main and eavesdropper channels have independent fading but at every fading block main channel has a stronger channel realization than the eavesdropper. In this case, a modified version of the aforementioned coding scheme is shown to achieve secrecy capacity.
\item Finally, we focus on the case of independent fading where eavesdropper's channel realization can be stronger than the main channel at a given block, but main channel is stronger on the average. In this scenario, we proposed an appropriate coding scheme, again based on hierarchical polar coding, and compare its performance with an upper bound on the secrecy capacity, which is obtained with instantaneous CSI assumption. We analyze the gap between the achievable rate and the outer bound, and investigate the role of instantaneous CSI in increasing secrecy rates for this model.
\end{itemize}

Overall, the proposed hierarchical polar coding scheme is a key component to achieve these results. This technique utilizes the polarization phenomena to convert the randomness in fading realizations (for both main and eavesdropper channels) to erasure channels over which additional polarization layer is used. Based on such a decomposition and injecting random symbols in appropriate positions (to achieve security over fading blocks), we establish a coding scheme that can secure messages without the need of any instantaneous CSI. (Detailed description of this technique is given later in the sequel.)

\subsection{Related Work}

In addition to \cite{Mahdavifar:Wiretap11, Koyluoglu:Secure12, Andersson:Nested10, Hof:Secrecy10}, recent studies on the design of polar coding schemes to achieve secrecy include \cite{Abbe:Key12, Sasoglu:Wiretap13, Sutter:Key13, Chou:Key13, Wei:Polar14, Chou:Polar14, Gulcu:Security14, Zheng:Polar14}, where strong security is considered in \cite{Koyluoglu:Secure12, Sasoglu:Wiretap13, Wei:Polar14, Gulcu:Security14}, key agreement/generation is studied in \cite{Koyluoglu:Secure12, Sutter:Key13, Chou:Key13}, and other channel models (different than discrete memoryless wiretap channel) are considered in \cite{Chou:Polar14, Zheng:Polar14, Andersson:Nested10}. Our model is similar to the fading models considered in \cite{Koyluoglu:Secure12, Zheng:Polar14}, but differentiates from all these prior studies in that only a statistical (i.e., distribution) CSI for \emph{both} receiver and eavesdropper channels is assumed at the transmitter. Polar coding schemes for fading wiretap channels are first studied in \cite{Koyluoglu:Secure12}, where the transmitter has the knowledge of instantaneous CSI for the receiver's channel and statistical CSI for the eavesdropper's channel. With this setup, a key agreement scheme is proposed based on utilizing polar codes for each fading block, where the communicated bits over fading blocks are then used in a privacy amplification step to construct secret keys. This technique when combined with invertible extractors allows for secure message transmission but with the requirement of receiver CSI at the transmitter\cite{Koyluoglu:Secure12}. Recent work \cite{Zheng:Polar14} proposes a polar coding scheme that utilizes artificial noise and multiple transmit antennas under the same assumption (instantaneous CSI for receiver and statistical CSI for eavesdropper) for the fading channels. However, a guarantee of secrecy rate with some probability (not the corresponding channel capacity) is achieved. In contrast, in this paper, we consider a fading channel model where the transmitter does not need to know any instantaneous CSI, but only its distribution for \emph{both} receiver and eavesdropper channels. The hierarchical polar coding scheme proposed in this paper, to the best of our knowledge, is the first provably secrecy capacity achieving coding scheme for fading (binary symmetric) wiretap channels. Considering that this type of binary channels model the AWGN channels with BPSK modulation and demodulation, our scheme covers a wide application scenarios in practice.

\subsection{Organization}

The rest of the paper is organized as follows. Section~\ref{sec:polar_intro} presents a brief introduction to polar coding scheme. Section~\ref{sec:simultaneous_fading} focuses on the simultaneous fading model, where the achievable coding scheme based on hierarchical polar coding is developed and the secrecy capacity is characterized. Section~\ref{sec:independent_fading} investigates the independent fading model, where modified coding schemes are developed specific to this model and the gap to capacity is analyzed via the difference between upper and lower bounds. Finally, Section~\ref{sec:conclusion} concludes the paper.

%%%%%%%%%%%%%%%%%%%%%%%%%%%%%%%%%%%%%%%%%%%%%%%%%%%%%%%%%%%%%%%%%%%%%%%%%%%%%%
%%%%%%%%%%%%%%%%%%%%%%%%%%%%%%%%%%%%%%%%%%%%%%%%%%%%%%%%%%%%%%%%%%%%%%%%%%%%%%
\section{Introduction to Polar Coding}
\label{sec:polar_intro}

Before moving to the main body of this paper, we first introduce the preliminary for polar coding. The construction of polar code is based on a phenomenon referred to as \emph{channel polarization}. Consider a binary-input discrete memoryless channel $\mathcal{W}_{\textrm{B-DMC}}:\mathcal{X}\to\mathcal{Y}$, where $\mathcal{X}=\{0,1\}$. Define
\begin{equation}
\bm{F}=\left[\begin{array}{cc}
  1 & 0 \\
  1 & 1
\end{array}\right].\nonumber\end{equation}
Let $\bm{B}_N$ be the bit-reversal operator as defined in \cite{Arikan:Channel08}, where $N$ is a power of $2$. By applying the transform $\bm{G}_N=\bm{B}_N\bm{F}^{\otimes \log N}$ ($\bm{F}^{\otimes \log N}$ denotes the $\log N$th Kronecker product of $\bm{F}$) to $u_{1:N}$, the encoded codeword given by $x_{1:N}=u_{1:N} \bm{G}_N$  is transmitted through $N$ independent copies of $\mathcal{W}_{\textrm{B-DMC}}$. Now, consider $N$ binary-input coordinate channels $\mathcal{W}_N^{(i)}:\mathcal{X}\to\mathcal{Y}^N\times\mathcal{X}^{i-1}$ ($i\in\{1,\ldots,N\}$), the transition probability is given by
\begin{equation}
\mathcal{W}_N^{(i)}(y_{1:N},u_{1:{i-1}}|u_i)\triangleq \sum_{u_{{i+1}:N}}\frac{1}{2^{N-1}}\mathcal{W}_{\textrm{B-DMC}}^N(y_{1:N}|u_{1:N}\bm{G}_N).\nonumber
\end{equation}
Remarkably, as $N\to\infty$, the channels $\mathcal{W}_N^{(i)}$ polarize to either noiseless or pure-noisy, and the fraction of noiseless channels is close to $I(\mathcal{W}_{\textrm{B-DMC}})$, the symmetric capacity of channel $\mathcal{W}_{\textrm{B-DMC}}$ \cite{Arikan:Channel08}.

\begin{figure}[t]
  \includegraphics[scale=0.8]{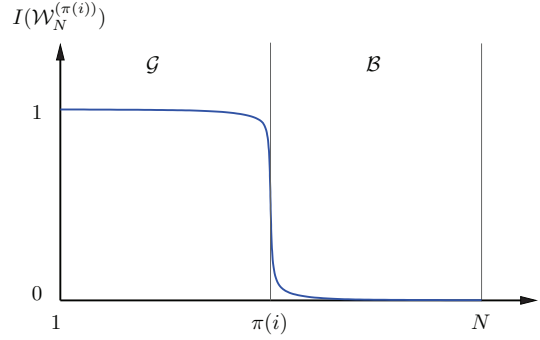}
  \centering
  \caption{Illustration of channel polarization for polar codes}\label{fig:channel_polarization}
\end{figure}

Given this polarization phenomenon (as shown in Fig.~\ref{fig:channel_polarization}), polar codes can be considered as $\bm{G}_N$-coset codes with parameters $(N,K,\mathcal{A},u_{\mathcal{A}^c})$, where $u_{\mathcal{A}^c}\in\mathcal{X}^{N-K}$ is frozen vector (can be set to all-zeros for binary symmetric channels \cite{Arikan:Channel08}), and the information set $\mathcal{A}$ is chosen as a $K$-element subset of $\{1,\ldots,N\}$ such that the Bhattacharyya parameters satisfy $Z(\mathcal{W}_N^{(i)})\leq Z(\mathcal{W}_N^{(j)})$ for all $i\in\mathcal{A}$ and $j\in\mathcal{A}^c$, i.e., $\mathcal{A}$ denotes the set of indices for \emph{good} channels (that are noiseless in the limit). We use permutations (namely, $\pi$ and $\phi$ in the sequel) to denote the increasing order of Bhattacharyya parameter values for the polarization of underlying channels. (For instance, for block length $N$, $\pi(1)$ gives the most reliable polarized channel index, and $\pi(N)$ gives the most unreliable one.)

A decoder for a polar code is the successive cancelation (SC) decoder, which gives an estimate $\hat{u}_{1:N}$ of $u_{1:N}$ given knowledge of $\mathcal{A}$, $u_{\mathcal{A}^c}$, and $y_{1:N}$ by computing
\begin{align}
\hat{u}_i \triangleq \left\{
\begin{array}{cl}
  1, & \text{if }i\in \mathcal{A}, \text{ and } \frac{\mathcal{W}_N^{(i)}(y_{1:N},\hat{u}_{1:{i-1}}|0)}{\mathcal{W}_N^{(i)}(y_{1:N},\hat{u}_{1:{i-1}}|1)}\geq 1,\\
    0, & \text{otherwise },
\end{array}
\right.\nonumber
\end{align}
in the order $i$ from $1$ to $N$.
It has been shown that, by adopting an SC decoder, polar codes achieve any rate $R<I(\mathcal{W}_{\textrm{B-DMC}})$ with a decoding error scaling as $O(2^{-N^{\beta}})$, where $\beta<1/2$. Moreover, the complexity for both encoding and decoding is $O(N\log N)$.

%%%%%%%%%%%%%%%%%%%%%%%%%%%%%%%%%%%%%%%%%%%%%%%%%%%%%%%%%%%%%%%%%%%%%%%%%%%%%%
%%%%%%%%%%%%%%%%%%%%%%%%%%%%%%%%%%%%%%%%%%%%%%%%%%%%%%%%%%%%%%%%%%%%%%%%%%%%%%
\section{Hierarchical Polar Coding for Simultaneous Fading Case}
\label{sec:simultaneous_fading}

\subsection{Problem setup}

In this section, we investigate the case where main channel and eavesdropper fade simultaneously. More precisely, consider the fading (binary symmetric) wiretap channel model (Fig.~\ref{fig:system}): Alice wishes to send message to Bob through the main channel $\mathcal{W}$, where the channel experiences the following block fading phenomenon: with probability $q_1$, channel $\mathcal{W}$ is a BSC$(p_1)$ (in the superior state), and with the rest probability $q_2\triangleq1-q_1$, channel $\mathcal{W}$ is a BSC$(p_2)$ (in the degraded state). On the same time, the transmission also reaches to an adversary (Eve) through the wiretap channel $\mathcal{W}^*$, where $\mathcal{W}^*$ is degraded compared to the main channel, and experiences the same fading state as the main channel. In particular, when $\mathcal{W}$ is a BSC$(p_1)$, $\mathcal{W}^*$ is a BSC$(p_1^*)$; when $\mathcal{W}$ is a BSC$(p_2)$, $\mathcal{W}^*$ is a BSC$(p_2^*)$. Under this system model, we have $p_1\leq p_2\leq 0.5$, $p_1^*\leq p_2^*\leq 0.5$, $p_1\leq p_1^*$, and $p_2\leq p_2^*$.
\begin{remark}
Simultaneous fading model considers the case where main channel and eavesdropper channel experience the same fading states (i.e., both are superior or both are degraded), and eavesdropper channel is assumed to be degraded to the main channel over each fading block. In Section~\ref{sec:independent_fading}, we consider an independent fading model where the two channels take their degraded/superior fading states independent of each other. (In this scenario, the eavesdropper can be stronger than the main channel for a given fading block, but the main channel is assumed to be stronger on the average.)
\end{remark}

\begin{figure}[t]
  \includegraphics[scale=0.8]{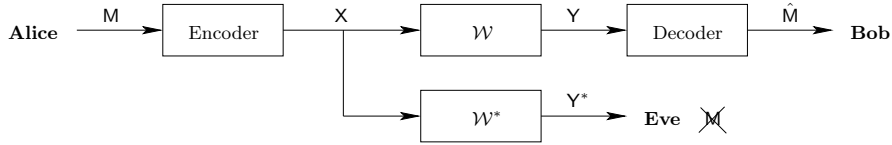}
  \centering
  \caption{System model for wiretap channels.}\label{fig:system}
\end{figure}

In general, fading coefficients vary at a much slower pace than the transmission symbol duration. For such cases, block fading model is considered, where the channel state is assumed to be constant within each coherence time interval, and follows a stationary ergodic process across fading blocks \cite{Tse:Wireless05}. To this end, we consider the practical scenario where channel state information (CSI) is available only at the decoder (CSI-D), while the encoder only knows the statistics of channel states. Under this model, a secrete message $\mathsf{M}$ is encoded by an encoding function $f(\cdot)$ to generate transmitted symbols: $\mathsf{X}_{1:NB}=f(\mathsf{M})$, where $N$ is the length of a fading block, and $B$ is the number of blocks. At the receiver, a decoding function $g(\cdot)$ gives an estimate of the estimate $\hat{\mathsf{M}}$, i.e., $\hat{\mathsf{M}}=g(\mathsf{Y}_{1:NB})$. The reliability of transmission is satisfied if
\begin{align}
P_{\textrm{e}}\triangleq\text{Pr}\{\mathsf{M}\neq\hat{\mathsf{M}}|\mathsf{Y}_{1:NB},\mathsf{S}\}\to 0, \textrm{ as } N,B \to\infty \label{equ:defn_reliability}
\end{align}
where $\mathsf{S}$ denotes CSI, and (weak) security is defined as achieving
\begin{align}
\frac{1}{NB} I(\mathsf{M};\mathsf{Z}_{1:NB}|\mathsf{S})\to 0,  \textrm{ as } N,B \to\infty.\label{equ:defn_security}
\end{align}
We denote the secrecy capacity of this model as $C^s_{\textrm{CSI-D}}$, which represents the highest achievable secrecy rate satisfying reliability \eqref{equ:defn_reliability} and secrecy \eqref{equ:defn_security} constraints.

\subsection{Upper bound to the secrecy capacity}

Under the degraded assumption, the secrecy capacity of the wiretap system can be upper bounded as reported in the following result.

\begin{lemma}\label{thm:UpperBound1}
The secrecy capacity for the simultaneous fading model is upper bounded by
$$C^s_{\textrm{CSI-D}}  \leq q_1[H(p_1^*)-H(p_1)]+q_2[H(p_2^*)-H(p_2)].$$
\end{lemma}

\begin{proof}
We have
\begin{align}
C^s_{\textrm{CSI-D}} & \triangleq \max_{p(x)} \left[I(\mathsf{X};\mathsf{Y}|\mathsf{S})-I(\mathsf{X};\mathsf{Z}|\mathsf{S})\right]\nonumber\\
                    & \overset{(a)}{\leq}\max_{p(x|s)} \left[I(\mathsf{X};\mathsf{Y}|\mathsf{S})-I(\mathsf{X};\mathsf{Z}|\mathsf{S})\right]\nonumber\\
                    & = \max_{p(x|1)}q_1[I(\mathsf{X};\mathsf{Y}|\mathsf{S}=1)-I(\mathsf{X};\mathsf{Z}|\mathsf{S}=1)]\nonumber\\
                    &\quad+ \max_{p(x|2)}q_2[I(\mathsf{X};\mathsf{Y}|\mathsf{S}=2)-I(\mathsf{X};\mathsf{Z}|\mathsf{S}=2)]\nonumber\\
                    & \overset{(b)}{=}q_1[H(p_1^*)-H(p_1)]+q_2[H(p_2^*)-H(p_2)],\label{equ:capacity_system}
\end{align}
where
\begin{itemize}
\item[$(a)$] follows by upper bounding the secrecy capacity with the case where encoder has CSI and adapts its coding scheme according to the channel states, i.e., $C^s_{\textrm{CSI-ED}}$ \cite{Gopala:Secrecy08};
\item[$(b)$] is due to the secrecy capacity result for the degraded binary symmetric wiretap channel \cite{Gamal:Network11}.
    \end{itemize}

\end{proof}

In this paper, assuming CSI is available only at the receivers, we provide a polar coding scheme that achieves this upper bound while satisfying reliability \eqref{equ:defn_reliability} and security \eqref{equ:defn_security} constraints. To this end, the upper bound \eqref{equ:capacity_system} gives the secrecy capacity of our model. For the moment, we assume $p_1\leq p_2\leq p_1^*\leq p_2^*$, and the remaining case ($p_1\leq p_1^*\leq p_2\leq p_2^*$) is detailed in Section~\ref{sec:simultaneous:case2}.

\begin{figure}[t]
  \includegraphics[scale=0.8]{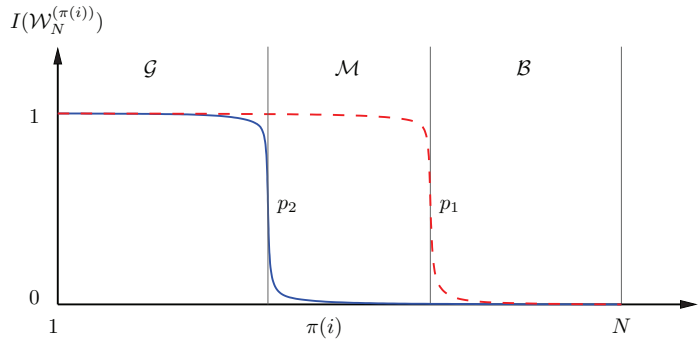}
  \centering
  \caption{Illustration of polarizations for two binary symmetric channels.}\label{fig:fading_polar}
\end{figure}

\subsection{Intuition behind the proposed coding scheme}

The intuition of hierarchical polar coding scheme originates from the polarization of degraded channels \cite{Korada:Thesis09}. When polarizing two binary symmetric channels BSC$(p_1)$ and BSC$(p_2)$ with $p_1\leq p_2$, the \emph{good} channels (i.e., noiseless as block length tends to infinity) of the polarized degraded channel BSC$(p_2)$ is a subset of that of the superior channel BSC$(p_1)$. As illustrated in Fig.~\ref{fig:fading_polar}, set $\mathcal{G}$ contains all good channel indices after permutation for both channels, while set $\mathcal{B}$ contains all bad channel indices after permutation for both channels. \cite{Si:Polar13} utilizes this property to construct hierarchical polar codes in order to achieve the capacity of fading binary symmetric channels. More precisely, polar codes are not only designed over channel uses within each fading block, but also utilized over different fading blocks. Inspired by this design, and combined with the polar coding scheme for wiretap channels \cite{Mahdavifar:Wiretap11, Koyluoglu:Secure12, Andersson:Nested10, Hof:Secrecy10}, we design the proposed polar coding scheme for fading binary symmetric wiretap channels.

\subsection{The scenario of $p_1\leq p_2\leq p_1^*\leq p_2^*$}
\label{sec:simultaneous:case1}

\begin{theorem}\label{thm:LowerBound1}
The secrecy capacity for the simultaneous fading model for $p_1\leq p_2\leq p_1^*\leq p_2^*$ is given by
$$C^s_{\textrm{CSI-D}}  = q_1[H(p_1^*)-H(p_1)]+q_2[H(p_2^*)-H(p_2)].$$
\end{theorem}

\begin{proof}
The upper bound follows from Lemma~\ref{thm:UpperBound1}. We detail the coding scheme (ie., a lower bound) as follows:

\textbf{Encoder:}

The encoder works in two phases (see Fig.~\ref{fig:encoder}), hierarchically using polar codes to generate an $NB$-length codeword.

\textbf{Phase I (BEC Encoding):}

Here, we consider two sets of messages to be encoded using polar encoders designed for binary erasure channels (BECs). For the first set of messages, we generate $|\mathcal{M}_1|$ number of BEC polar codes, where
\begin{align}
|\mathcal{M}_1|=N[H(p_2^*)-H(p_1^*)].\label{equ:length_M1}
\end{align}
Consider a set of blockwise messages $u^{(i)}_{1:|\mathcal{A}^c|}$ with $i\in\{1,\ldots,|\mathcal{M}_1|\}$, where $\mathcal{A}$ is the information set for BEC$(q_2)$, i.e.,
\begin{align}
&|\mathcal{A}|=B\cdot q_1,\label{equ:length_A}\\
&|\mathcal{A}^c|=B\cdot q_2.\label{equ:langth_Ac}
\end{align}
For every $u^{(i)}_{1:|\mathcal{A}^c|}$, we combine it with $|\mathcal{A}|$ random bits to construct polar codeword $\tilde{u}^{(i)}_{1:B}$.
Denoting the permutation for BEC$(q_2)$ channel as $\phi$, and the uniform random string as $r^{(i)}_{1:|\mathcal{A}|}$ (each bit is Ber$(1/2)$ distributed), the encoding process is given by
\begin{align}
\tilde{u}^{(i)}_{1:B}&=\bm{\mu}^{(i)}_{1:B}\times \bm{G}_B,\nonumber\\
\phi\left(\bm{\mu}^{(i)}_{1:B}\right)&=\left[\begin{array}{ccc}
r^{(i)}_{1:|\mathcal{A}|}& | & u^{(i)}_{1:|\mathcal{A}^c|}
\end{array}\right],\nonumber%\label{equ:definition_mu}
\end{align}
for every $i\in\{1,\ldots,|\mathcal{M}_1|\}$, where $\bm{G}_B$ is the polar generator matrix with size $B$. By collecting all $\tilde{u}^{(i)}_{1:B}$ together, the encoder generates a $|\mathcal{M}_1|\times B$ matrix $\tilde{\bm{U}}$. We denote $\tilde{\bm{U}}^T_{k}$ as the $k$-th row of the transpose of $\tilde{\bm{U}}$, where $k\in\{1,\ldots,B\}$.

Secondly, we generate $|\mathcal{M}_2|$ number of BEC polar codes, where
\begin{align}
|\mathcal{M}_2|=N[H(p_2)-H(p_1)].\label{equ:length_M2}
\end{align}
Consider another set of blockwise messages $v^{(j)}_{1:|\mathcal{A}|}$ with $j\in\{1,\ldots,|\mathcal{M}_2|\}$. Each message is set as information bits to construct polar codeword $\tilde{v}^{(j)}_{1:B}$.
More formally, this encoding process is given by
\begin{align}
\tilde{v}^{(j)}_{1:B}&=\bm{\nu}^{(j)}_{1:B}\times \bm{G}_B,\nonumber \\ %\label{equ:encoder_BEC2}
\phi\left(\bm{\nu}^{(j)}_{1:B}\right)&=\left[\begin{array}{ccc}
v^{(j)}_{1:|\mathcal{A}|}  & | & 0
\end{array}\right],\nonumber%\label{equ:definition_nu}
\end{align}
for every $j\in\{1,\ldots,|\mathcal{M}_2|\}$.
The collection of all $\tilde{v}^{(j)}_{1:B}$ together is denoted as a $|\mathcal{M}_2|\times B$ matrix $\tilde{\bm{V}}$. We denote $\tilde{\bm{V}}^T_{k}$ as the $k$-th row of the transpose of $\tilde{\bm{V}}$, where $k\in\{1,\ldots,B\}$.

\begin{figure}[t!]
  \includegraphics[scale=0.78]{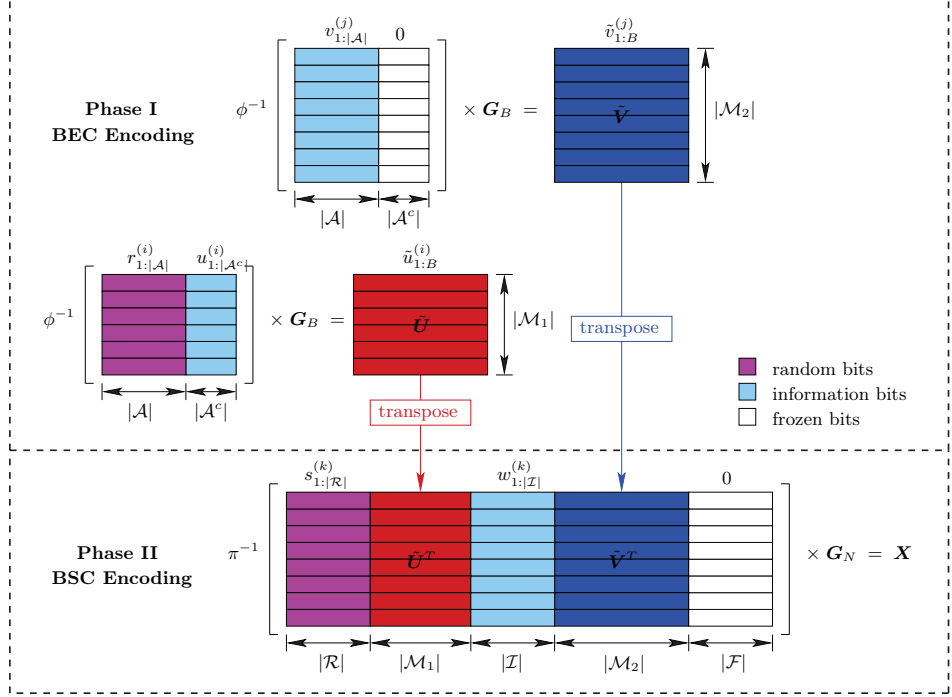}
  \centering
  \caption{Encoder of the polar coding scheme for fading wiretap channels.}\label{fig:encoder}
\end{figure}

\textbf{Phase II (BSC Encoding):}

In this phase, we generate $B$ number of BSC polar codes, each with length $N$. The encoded codewords from previous phase are embedded as messages of this phase. We consider a set of messages $w^{(k)}_{1:|\mathcal{I}|}$ with $k\in\{1,\ldots,B\}$, where
\begin{align}
|\mathcal{I}|=N[H(p_1^*)-H(p_2)].\label{equ:length_I}
\end{align}
For every $w^{(k)}_{1:|\mathcal{I}|}$, we introduce random bits $s_{1:|\mathcal{R}|}^{(k)}$, where
\begin{align}
|\mathcal{R}|=N[1-H(p_2^*)],\label{equ:length_R}
\end{align}
and combine the output from the previous phase as message to construct polar codeword $x^{(k)}_{1:N}$.
More formally, if we denote the reordering permutation for BSC$(p_1)$ as $\pi$, then the encoder of this phase can be expressed as
\begin{align}
x^{(k)}_{1:N}&=\bm{\omega}_{1:N}^{(k)}\times \bm{G}_N,  \nonumber\\%\label{equ:encoder_BSC}
\pi\left(\bm{\omega}_{1:N}^{(k)}\right)&=\left[\begin{array}{ccccccccc}
s^{(k)}_{1:|\mathcal{R}|} & | & \tilde{\bm{U}}^T_{k} & | & w^{(k)}_{1:|\mathcal{I}|} & | & \tilde{\bm{V}}^T_{k} & | & 0
\end{array}\right],\nonumber%\label{equ:definition_omega}
\end{align}
for every $k\in \{1,\ldots, B\}$, where $\bm{G}_N$ is the polar generator matrix with size $N$. That is, the codewords generated from BEC encoding phase are transposed and embedded into the messages of the BSC encoding process. We denote these codewords by a $B\times N$ matrix $\bm{X}$. The proposed encoder is illustrated in Fig.~\ref{fig:encoder}.

\textbf{Decoder for the Main Channel:}

The codewords $x^{(k)}_{1:N}$ are transmitted through both the main channel and the wiretap channel. After receiving the output sequence $y^{(k)}_{1:N}$ for all $k\in\{1,\ldots,B\}$, the task of the decoder at Bob is to make estimates for all the information and random bits.
In particular, the decoder aims to recover $u^{(i)}_{1:|\mathcal{A}^c|}$, $v^{(j)}_{1:|\mathcal{A}|}$, $w^{(k)}_{1:|\mathcal{I}|}$, $r^{(i)}_{1:|\mathcal{A}|}$, and $s^{(k)}_{1:|\mathcal{R}|}$ successfully with high probability.
As that of the encoding process, the decoding process also works in phases (see Fig.~\ref{fig:decoder_main}).

\begin{figure*}[t]
  \includegraphics[scale=0.78]{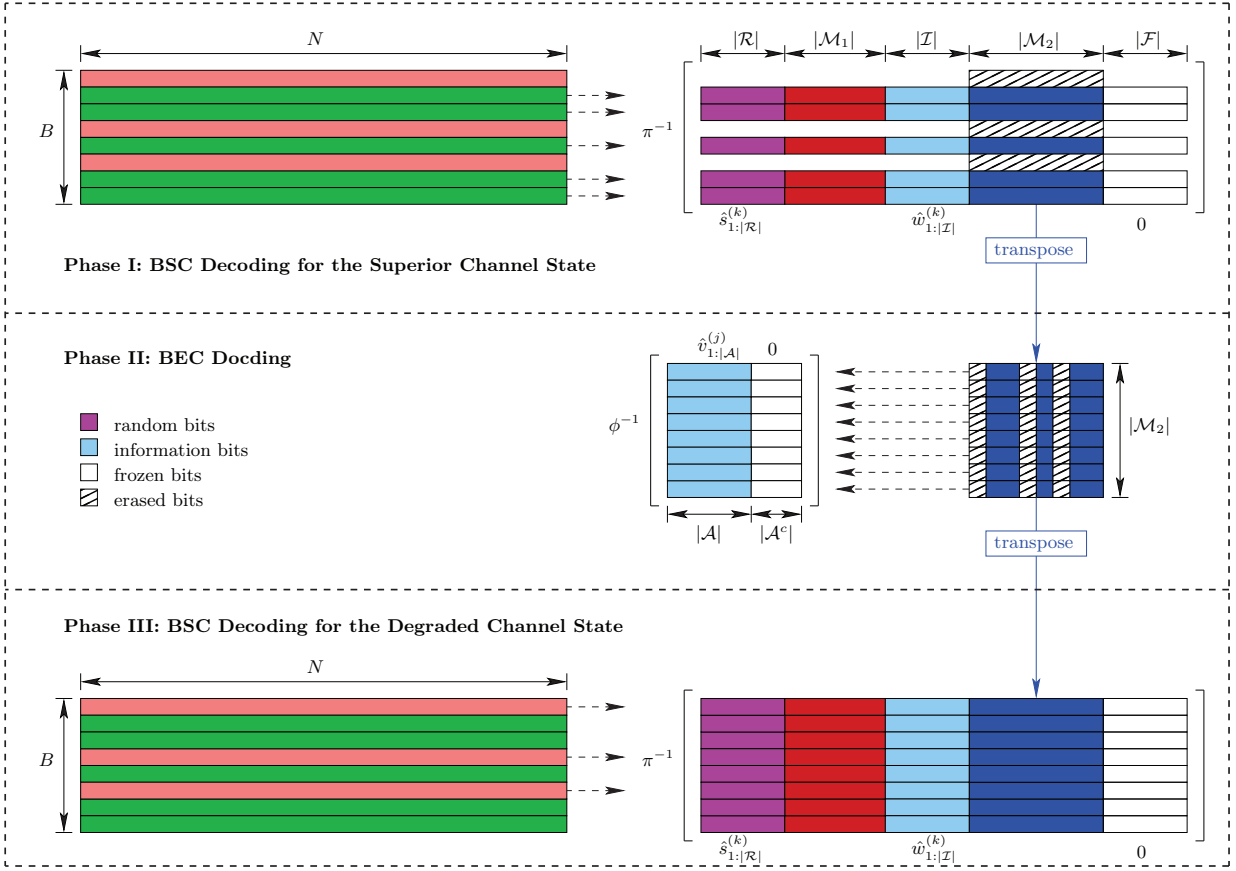}
  \centering
  \caption{Decoder at the main channel receiver given the knowledge of the channel states information.}\label{fig:decoder_main}
\end{figure*}

\textbf{Phase I (BSC Decoding for the Superior Channel State):}

In this phase, using the BSC SC decoder, channels corresponding to the superior state are decoded. More precisely, since the receiver knows the channel states, it can adopt the correct SC decoder to obtain estimates $\hat{\bm{\omega}}^{(k)}_{1:N}$ from $y^{(k)}_{1:N}$ for every $k$ corresponding to the superior channel state. To this end, the decoder adopted in this phase is the classical BSC SC polar decoder with parameter $p_1$, i.e.,
\begin{align}
\hat{\bm{\omega}}^{(k)}_n=\left\{\begin{array}{ll}
1,  & \text{if } n\notin\mathcal{F}\text{, and }\frac{\mathcal{W}^{(n)}_{1,N}(y^{(k)}_{1:N},\hat{\bm{\omega}}^{(k)}_{1:n-1}|1)}{\mathcal{W}^{(n)}_{1,N}(y^{(k)}_{1:N},\hat{\bm{\omega}}^{(k)}_{1:n-1}|0)}\geq 1,\\
0,  & \text{otherwise},
\end{array}
\right.\nonumber
\end{align}
in the order $n$ from $1$ to $N$, and $\mathcal{W}^{(n)}_{1,N}$ is the $n$-th polarized channel from BSC$(p_1)$. Then, for every $k$ corresponding to the superior channel state, the decoder can obtain the messages (with the knowledge of the frozen symbols corresponding to $\mathcal{F}$ indices)
\begin{align}
\pi\left(\hat{\bm{\omega}}^{(k)}_{1:N}\right)=\left[ \begin{array}{ccccccccc}
\hat{s}^{(k)}_{1:|\mathcal{R}|} & | & \hat{\tilde{\bm{U}}}^T_{k} & | & \hat{w}^{(k)}_{1:|\mathcal{I}|} & | & \hat{\tilde{\bm{V}}}^T_{k} & | & 0
\end{array} \right]. \nonumber%\label{equ:decode_omega_superior}
\end{align}
However, for the blocks with degraded channel states, one cannot decode reliably because the frozen bits corresponding to set $\mathcal{M}_2$ are unknown at the decoder. At this point, we use the next phase to decode these frozen bits using a BEC SC decoder. To proceed, we construct a $B\times |\mathcal{M}_2|$ matrix $\hat{\tilde{\bm{V}}}^T$ such that its rows corresponding to the superior state are determined in previous decoding process, while the ones corresponding to the degraded states are all set to erasures.

\textbf{Phase II (BEC Decoding):}

In this phase, we decode the frozen bits with respect to the degraded channel state. More precisely, each row of matrix $\hat{\tilde{\bm{V}}}$, denoted by $\hat{\tilde{\bm{V}}}_j$ for $j\in\{1,\ldots,|\mathcal{M}_2|\}$, is considered as the input to the decoder, and the receiver aims to obtain an estimate of the information bits from it using BEC SC decoder. To this end, the decoder adopted in this phase is the classical BEC SC decoder with parameter $q_2$, i.e.,
\begin{align}
\hat{\bm{\nu}}^{(j)}_b=\left\{\begin{array}{ll}
1,  & \text{if } b\in\mathcal{A}\text{, and }\frac{\mathcal{W}^{(b)}_{\text{e},B}(\hat{\tilde{\bm{V}}}_j,\hat{\bm{\nu}}^{(j)}_{1:b-1}|1)}{\mathcal{W}^{(b)}_{\text{e},B}(\hat{\tilde{\bm{V}}}_j,\hat{\bm{\nu}}^{(j)}_{1:b-1}|1)}\geq 1,\\
0,  & \text{otherwise},
\end{array}
\right.\nonumber
\end{align}
in the order $b$ from $1$ to $B$, and $\mathcal{W}^{(b)}_{\text{e},B}$ is the $b$-th polarized channel from BEC$(q_2)$. Then, for every $j$, the decoder can declare
\begin{align}
\phi\left(\hat{\bm{\nu}}^{(j)}_{1:B}\right)=\left[\begin{array}{ccc}
\hat{v}^{(j)}_{1:|\mathcal{A}|}& | & 0
\end{array}\right]. \nonumber%\label{equ:decode_nu}
\end{align}
At this point, the decoder can reconstruct all erased bits as well. More precisely, the erased rows in $\hat{\tilde{\bm{V}}}^T$ can be recovered, and they can be further utilized to decode the information bits in blocks with the degraded channel state in the next phase.

\textbf{Phase III (BSC Decoding for the Degraded Channel State):}

In this phase, the remaining blocks from Phase I are decoded by using BSC SC decoders with respect to degraded channel states. In particular, bits in the frozen set for the degraded channel state (set $\mathcal{F}$ and set $\mathcal{M}_2$) are known due to the previous phases. Hence, the receiver can decode from $y^{(k)}_{1:N}$ using BSC SC decoder with parameter $p_2$, i.e.,
\begin{align}
\hat{\bm{\omega}}^{(k)}_n=\left\{\begin{array}{ll}
1,  & \text{if } n\notin\mathcal{F}\text{, }n\notin\mathcal{M}_2 \text{,}\text{ and }  \frac{\mathcal{W}^{(n)}_{2,N}(y^{(k)}_{1:N},\hat{\bm{\omega}}^{(k)}_{1:n-1}|1)}{\mathcal{W}^{(n)}_{2,N}(y^{(k)}_{1:N},\hat{\bm{\omega}}^{(k)}_{1:n-1}|0)}\geq 1,\\
\hat{\tilde{\bm{V}}}^T_{kn}, &\text{if } n\in\mathcal{M}_2, \\
0,  & \text{otherwise},
\end{array}
\right.\nonumber
\end{align}
in the order $n$ from $1$ to $N$, and $\mathcal{W}^{(n)}_{2,N}$ is the $n$-th polarized channel from BSC$(p_2)$. Then, for every $k$ corresponding to the degraded channel state, the decoder declares
\begin{align}
\pi\left(\hat{\bm{\omega}}^{(k)}_{1:N}\right)=\left[ \begin{array}{ccccccccc}
\hat{s}^{(k)}_{1:|\mathcal{R}|} & | & \hat{\tilde{\bm{U}}}^T_{k} & | & \hat{w}^{(k)}_{1:|\mathcal{I}|} & | & \hat{\tilde{\bm{V}}}^T_{k} & | & 0
\end{array} \right]. \nonumber%\label{equ:decode_omega_degraded}
\end{align}
Hence, after this decoding procedure, the receiver makes an estimate $\hat{\tilde{\bm{U}}}$ of matrix $\tilde{\bm{U}}$, which further implies all information bits in $u^{(i)}_{1:|\mathcal{A}^c|}$ are decoded. Note that, in addition to information bits, all random bits are decoded reliably at Bob as well. However, in order to guarantee security, we set these bits random (instead of information).

\textbf{Achievable Rate and Reliability:}

The proposed hierarchical scheme allows for recovering all information bits (represented by light blue in Fig.~\ref{fig:encoder}) reliably, as long as the designed rates of polar codes do not exceed the corresponding channel capacities. Hence, the achievable rate is given by
\begin{align}
R   &=\frac{1}{NB}\left(|\mathcal{M}_2|\times|\mathcal{A}| +|\mathcal{M}_1|\times|\mathcal{A}^c|+ B\times|\mathcal{I}|\right)\nonumber\\
    &=[H(p_2)-H(p_1)]\times q_1 +[H(p_2^*)-H(p_1^*)]\times q_2+[H(p_1^*)-H(p_2)]\nonumber\\
    &=[H(p_1^*)-H(p_1)]\times q_1 +[H(p_2^*)-H(p_2)]\times q_2,\label{equ:achievable_rate}
\end{align}
where we have used \eqref{equ:length_M1}, \eqref{equ:length_A}, \eqref{equ:langth_Ac}, \eqref{equ:length_M2}, and \eqref{equ:length_I}. In this scheme, $B$ number of $N$-length polar codes are decoded in Phase I and III in total, and $|\mathcal{M}_2|$ number of $B$-length polar codes are decoded in Phase II. Hence, the decoding error probability is upper bounded by
\begin{align}
\text{Pr}\{\mathsf{M}\neq \hat{\mathsf{M}}|\mathsf{Y}_{1:NB},\mathsf{S}\}\leq B\cdot2^{-N^{\beta}}+|\mathcal{M}_2|\cdot2^{-B^{\beta}},\label{equ:error_bound_main}
\end{align}
where $\beta<1/2$; and, $\mathsf{M}$ is the collection of random variables representing for all information bits (its realizations include $u_{1:|\mathcal{A}^c|}^{(i)}$, $v_{1:|\mathcal{A}|}^{(j)}$, and $w_{1:|\mathcal{I}|}^{(k)}$), and $\hat{\mathsf{M}}$ is the estimate of $\mathsf{M}$ obtained by Bob. Noting that the right hand side of \eqref{equ:error_bound_main} tends to $0$ when implemented with properly large $B$ and $N$, the proposed scheme achieves the upper bound given by \eqref{equ:capacity_system} reliably.

\textbf{Security:}

Assume that, in addition to ${z}^{(k)}_{1:N}$, a genie reveals Eve all information bits $u^{(i)}_{1:|\mathcal{A}^c|}$, $v^{(j)}_{1:|\mathcal{A}|}$, and $w^{(k)}_{1:|\mathcal{I}|}$. Under this condition, we show that all random bits can be reliably decoded at Eve. More precisely, the decoder designed for the eavesdropper also works in phases, similar to the one for the main channel (see Fig.~\ref{fig:decoder_wiretap}).
\begin{itemize}
\item Phase I (BSC Decoding for the Superior Channel State): The decoder still works over the blocks with the superior channel state. However, for the eavesdropper channel with superior channel state, the frozen set consists of bits not only in set $\mathcal{F}$, but also in sets $\mathcal{M}_2$ and $\mathcal{I}$. Since we have assumed the information bits are known at Eve, the classical BSC$(p_1^*)$ SC decoder can be used to decode the random bits.
\item Phase II (BEC Decoding): This phase aims to recover the unknown frozen bits corresponding to the degraded channel state, where a similar scheme as that of the main receiver is adopted. More precisely, by modeling the appropriate symbols corresponding to degraded channel states as erasures, we utilize the BEC$(q_2)$ SC decoder over each row of the matrix after transpose. This scheme successively recovers the erased elements, as the frozen bits for this BEC is the information bits $u^{(i)}_{1:|\mathcal{A}^c|}$ and they are assumed to be known.
\item Phase III (BSC Decoding for the Degraded Channel State): Finally, the decoded result from the previous phase is utilized at the BSC decoding for the degraded state, where the classical BSC$(p_2^*)$ SC decoder is adopted.
\end{itemize}

\begin{figure*}[t]
  \includegraphics[scale=0.78]{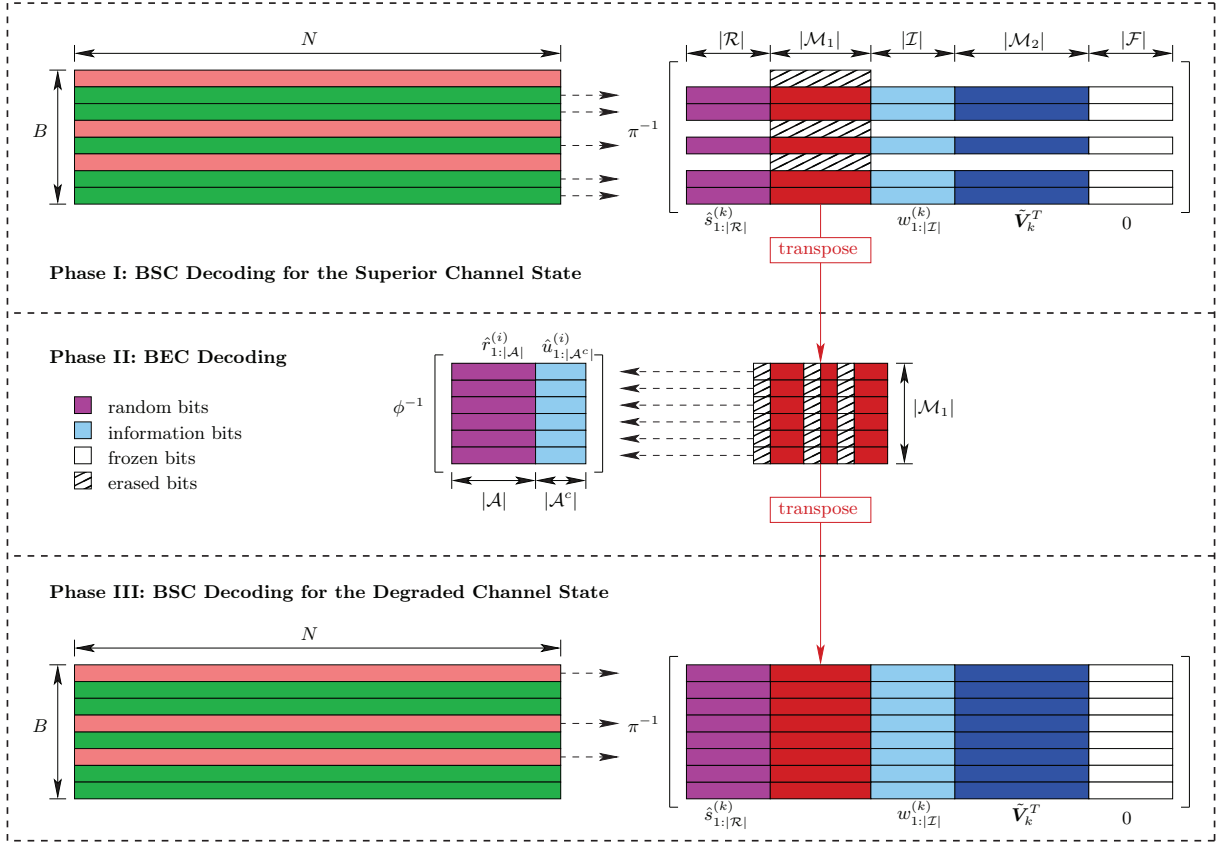}
  \centering
  \caption{Decoder at the eavesdropper given the knowledge of the channel states information and information bits.}\label{fig:decoder_wiretap}
\end{figure*}

By adopting this hierarchical polar decoder, Eve can decode all random bits with high probability, i.e.,
\begin{align}
\text{Pr}\{\mathsf{R}\neq \hat{\mathsf{R}}|\mathsf{Z}_{1:NB},\mathsf{M},\mathsf{S}\}\leq B\cdot2^{-N^{\beta}}+|\mathcal{M}_1|\cdot2^{-B^{\beta}}, \label{equ:error_bound_wiretap}
\end{align}
where $\mathsf{R}$ is the collection of random variables representing for random bits (its realization include $r^{(i)}_{1:|\mathcal{A}|}$ and $s^{(k)}_{1:|\mathcal{R}|}$), and $\hat{\mathsf{R}}$ is the estimate of $\mathsf{R}$. Then, using Fano's inequality, together with \eqref{equ:error_bound_wiretap}, we have
\begin{align}
&H(\mathsf{R}|\mathsf{Z}_{1:NB},\mathsf{M},\mathsf{S})\nonumber\\
&\quad\leq [B\cdot2^{-N^{\beta}}+|\mathcal{M}_1|\cdot2^{-B^{\beta}}]\cdot[|\mathcal{R}|\cdot B+|\mathcal{A}|\cdot|\mathcal{M}_1|]\nonumber\\
&\quad\quad\quad + H(B\cdot2^{-N^{\beta}}+|\mathcal{M}_1|\cdot2^{-B^{\beta}}).\label{equ:Fano_inequality}
\end{align}

Based on this, the following steps provide an upper bound (omitting the subscript of $\mathsf{Z}$):
\begin{align}
I(\mathsf{M};\mathsf{Z}|\mathsf{S})
&=I(\mathsf{M},\mathsf{R};\mathsf{Z}|\mathsf{S})-\left[H(\mathsf{R}|\mathsf{M},\mathsf{S})-H(\mathsf{R}|\mathsf{Z},\mathsf{M},\mathsf{S})\right]\nonumber\\
&\overset{\text{(a)}}{=}I(\mathsf{M},\mathsf{R};\mathsf{Z}|\mathsf{S})-H(\mathsf{R})+H(\mathsf{R}|\mathsf{Z},\mathsf{M},\mathsf{S})\nonumber\\
&\overset{\text{(b)}}{\leq} NB\cdot C_{\textrm{CSI-D}}(\mathcal{W}^*)-H(\mathsf{R})+H(\mathsf{R}|\mathsf{Z},\mathsf{M},\mathsf{S})\nonumber\\
&\overset{\text{(c)}}{=} NB\cdot C_{\textrm{CSI-D}}(\mathcal{W}^*)-|\mathcal{A}|\cdot|\mathcal{M}_1|-B\cdot|\mathcal{R}|+H(\mathsf{R}|\mathsf{Z},\mathsf{M},\mathsf{S})\nonumber\\
&\overset{\text{(d)}}{=} NB\cdot C_{\textrm{CSI-D}}(\mathcal{W}^*)-Bq_1\cdot N[H(p_2^*)-H(p_1^*)]-B\cdot N[1-H(p_2^*)]+H(\mathsf{R}|\mathsf{Z},\mathsf{M},\mathsf{S})\nonumber\\
&= NB\cdot C_{\textrm{CSI-D}}(\mathcal{W}^*)-NB\cdot q_1[1-H(p_1^*)]-NB\cdot q_2[1-H(p_2^*)]+H(\mathsf{R}|\mathsf{Z},\mathsf{M},\mathsf{S})\nonumber\\
&\overset{\text{(e)}}{=}H(\mathsf{R}|\mathsf{Z},\mathsf{M},\mathsf{S}),\nonumber
\end{align}
where
\begin{itemize}
\item[(a)] follows as $\mathsf{R}$ is independent of $\mathsf{M}$ and $\mathsf{S}$;
\item[(b)] is due to the definition of channel $\mathcal{W}^*$'s capacity with CSI-D;
\item[(c)] is due to the assumption that $\mathsf{R}$ is uniform;
\item[(d)] is due to equations \eqref{equ:length_M1}, \eqref{equ:length_A}, and \eqref{equ:length_R};
\item[(e)] is due to the ergodic capacity of the degraded fading eavesdropper channel with channel state information known only at the decoder \cite{Gamal:Network11}, i.e.,
$$C_{\textrm{CSI-D}}(\mathcal{W}^*)=q_1[1-H(p_1^*)]+q_2[1-H(p_2^*)].$$
\end{itemize}

Finally, combining with \eqref{equ:Fano_inequality}, we have
\begin{align}
\frac{1}{NB}I(\mathsf{M};\mathsf{Z}_{1:NB}|\mathsf{S})\to 0,\nonumber
\end{align}
as $N$ and $B$ tends to infinity (with proper choice of the their scaling relationship).
Hence, the proposed scheme achieves the secrecy constraint.

\end{proof}

\subsection{The Scenario of $p_1\leq p_1^*\leq p_2\leq p_2^*$}
\label{sec:simultaneous:case2}

In this section, we extend the aforementioned coding scheme to the scenario of $p_1\leq p_1^*\leq p_2\leq p_2^*$. Combined with Theorem~\ref{thm:LowerBound1} provided earlier in this section, this completes the proof for all possible cases of simultaneous fading and establishes the following result.

\begin{theorem}\label{thm:LowerBound2}
The secrecy capacity for the simultaneous fading model is given by
$$C^s_{\textrm{CSI-D}}  = q_1[H(p_1^*)-H(p_1)]+q_2[H(p_2^*)-H(p_2)].$$
\end{theorem}

\begin{proof}
Note that although $p_1^*\leq p_2$, the main channel is still stronger than the eavesdropper channel in each fading block (because of the simultaneous fading assumption). To this end, the upper bound reported in Lemma~\ref{thm:UpperBound1} still holds for this scenario. It remains to show the achievability for $p_1\leq p_1^*\leq p_2\leq p_2^*$.

\textbf{Encoding:}

From the previous scenario, the key idea for hierarchical polar coding scheme is setting the size of random bits be $NB\cdot C_{\textrm{CSI-D}}(\mathcal{W}^*)$ and setting the size
of information bits be $NB\cdot C^s_{\textrm{CSI-D}}(\mathcal{W})$. Based on this observation, the encoder for the scenario discussed here is illustrated in Fig.~\ref{fig:encoder_caseii}. Note that we still have five categories for channel indices after polarization. $\mathcal{R}$ and $\mathcal{F}$ remain the same as the previous scenario, but we do not have pure information set in this scenario due to $p_1^*\leq p_2$. Instead, a new set $\mathcal{M}_3$ contains coding results from random bits and frozen bits. More precisely, parameters shown in the figure are defined as follow:
\begin{align}
&|\mathcal{R}|=N[1-H(p_2^*)],\nonumber\\
&|\mathcal{M}_1|=N[H(p_2^*)-H(p_2)],\nonumber\\
&|\mathcal{M}_2|=N[H(p_1^*)-H(p_1)],\nonumber\\
&|\mathcal{M}_3|=N[H(p_2)-H(p_1^*)],\nonumber\\
&|\mathcal{F}|=N\cdot H(p_1),\nonumber\\
&|\mathcal{A}|=B\cdot q_1,\nonumber\\
&|\mathcal{A}^c|=B\cdot q_2.\nonumber
\end{align}

\begin{figure*}[t]
  \includegraphics[scale=0.78]{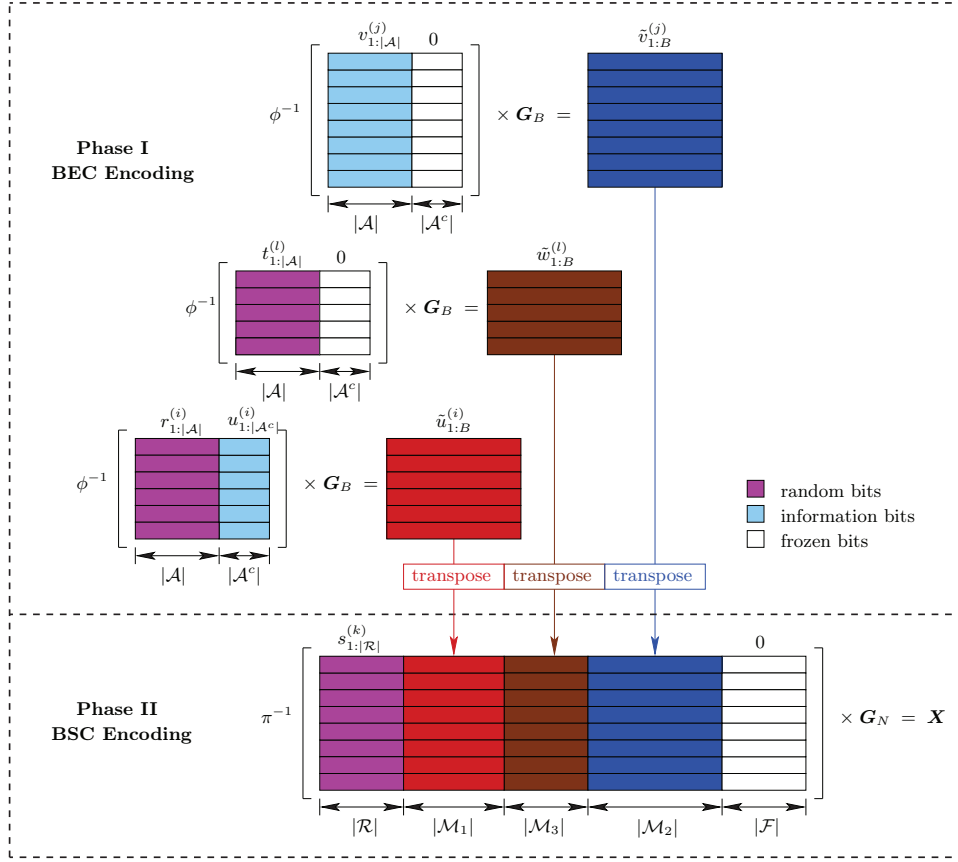}
  \centering
  \caption{Encoder for the scenario of $p_1\leq p_1^*\leq p_2\leq p_2^*$.}\label{fig:encoder_caseii}
\end{figure*}

Then, the encoding procedure works analog to the previous scenario, except that three sets of BEC encoding are performed and the resulting codewords are transposed and embedded into the second phase. In particular, the sketch of hierarchical coding scheme is described as follow:

\textbf{Phase I (BEC Encoding):}

\begin{itemize}
\item Random bits $r_{1:|\mathcal{A}|}^{(i)}$ combined with information bits $u_{1:|\mathcal{A}^c|}^{(i)}$ are encoded to generate $\tilde{u}_{1:B}^{(i)}$, for each $i\in{1,\ldots,|\mathcal{M}_1|}$;
\item Information bits $v_{1:|\mathcal{A}|}^{(j)}$ combined with frozen bits $0$ are encoded to generate $\tilde{v}_{1:B}^{(j)}$, for each $j\in{1,\ldots,|\mathcal{M}_2|}$;
\item Random bits $t_{1:|\mathcal{A}|}^{(l)}$ combined with frozen bits $0$ are encoded to generate $\tilde{w}_{1:B}^{(l)}$, for each $l\in{1,\ldots,|\mathcal{M}_3|}$.
\end{itemize}

\textbf{Phase II (BSC Encoding):}

Coded bits from Phase I are combined with random bits $s^{(k)}_{1:|\mathcal{R}|}$ and frozen bits $0$ are encoded to generate $x^{(k)}_{1:N}$, for each $k\in{1,\ldots,B}$.

\begin{figure*}[t]
  \includegraphics[scale=0.78]{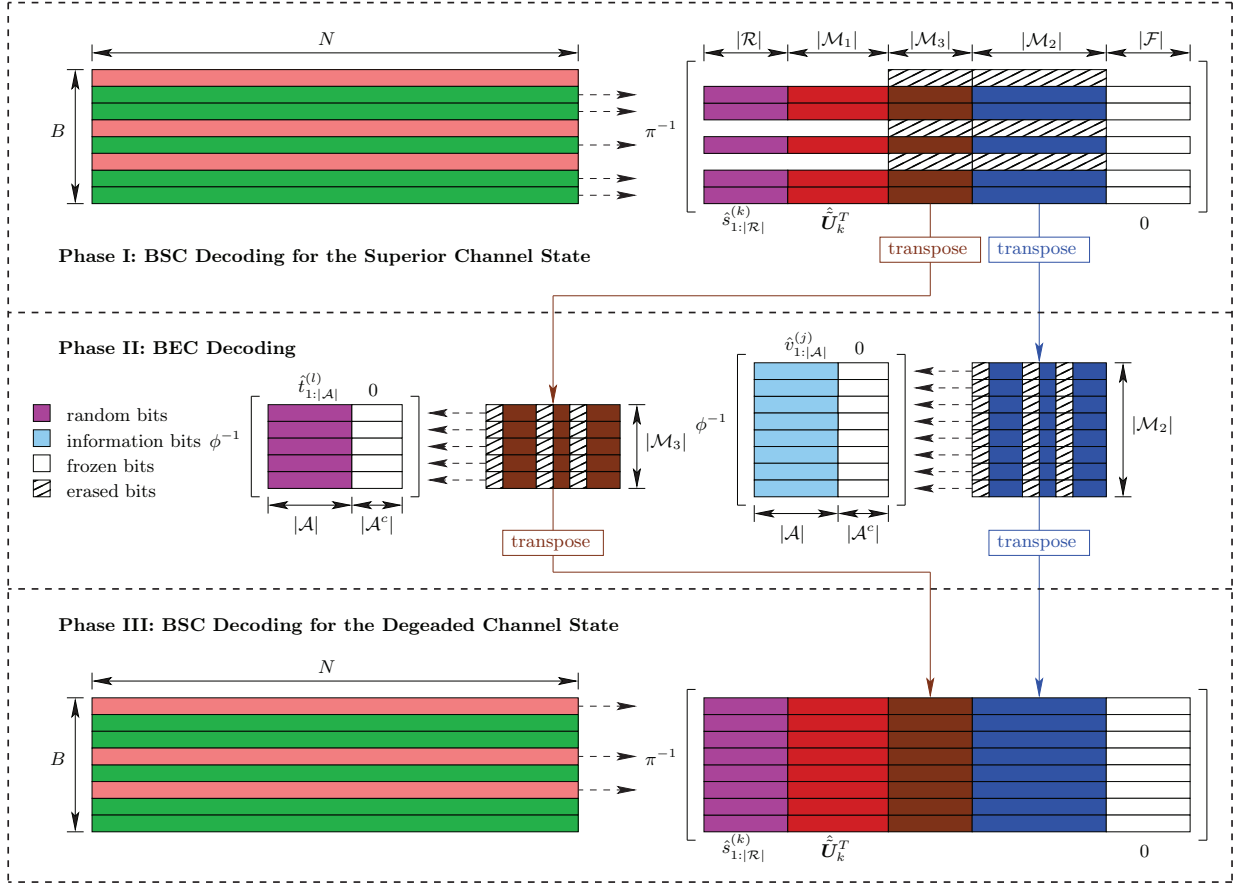}
  \centering
  \caption{Decoder at the main channel for the scenario of $p_1\leq p_1^*\leq p_2\leq p_2^*$.}\label{fig:decoder_main_caseii}
\end{figure*}

\textbf{Decoder for the Main Channel:}

The decoder at the main channel also works in phases. Quite similar to the previous case, the sketch of decoder is as follows (illustrated in Fig.~\ref{fig:decoder_main_caseii}):
\begin{itemize}
\item Phase I (BSC Decoding for the Superior State): Decode the block with respect to the superior state using BSC$(p_1)$ SC decoder by choosing frozen bits as $0$.
\item Phase II (BEC Decoding): Add erasures to the decoded bits in set $\mathcal{M}_3$ and $\mathcal{M}_2$ from previous phase, then decode both the random bits and information bits using BEC$(q_2)$ SC decoder by choosing frozen bits as $0$.
\item Phase III (BSC Decoding for the Degraded State): Recover all bits in set $\mathcal{M}_3$ and $\mathcal{M}_2$ to make them the frozen bits, and decode the block with respect to the degraded state using BSC$(p_2)$ SC decoder.
\end{itemize}

\begin{figure*}[t]
  \includegraphics[scale=0.78]{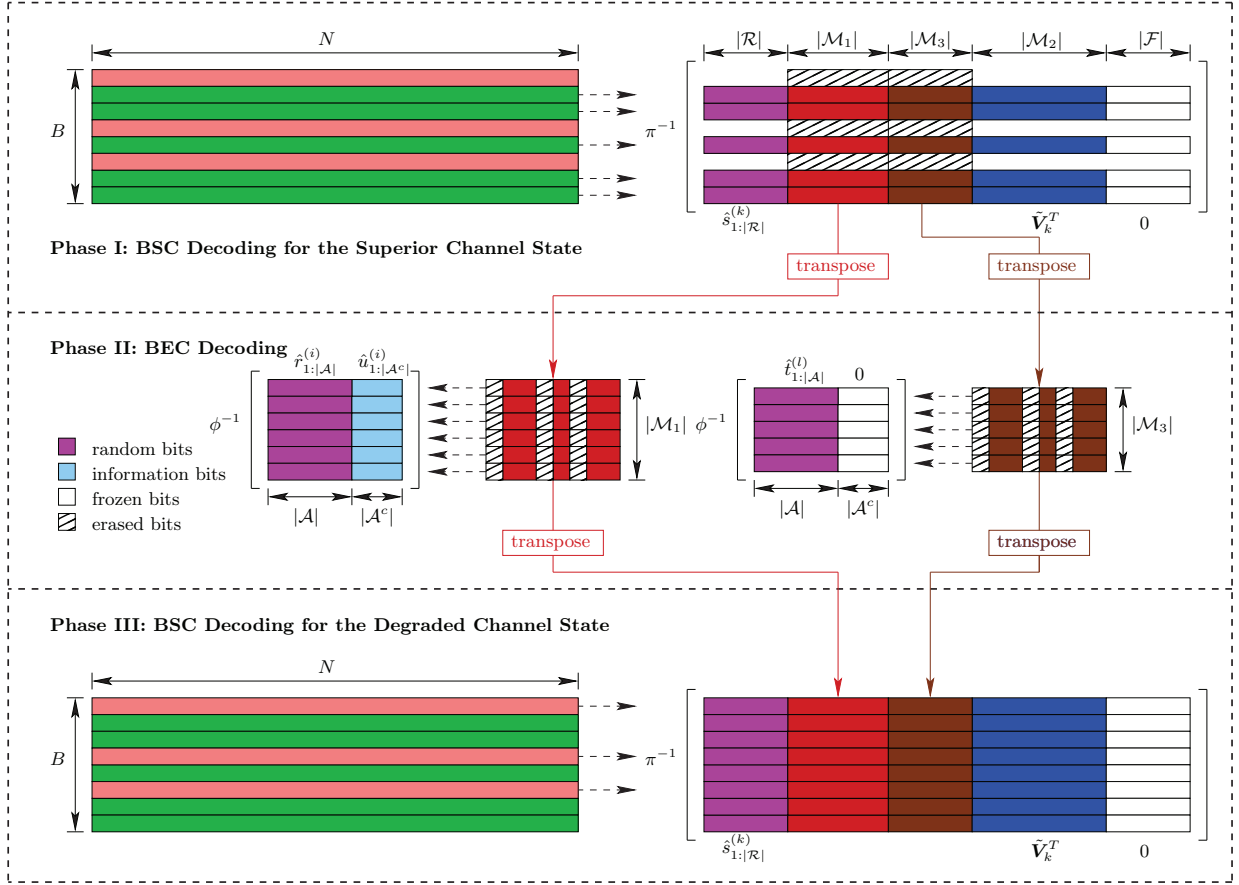}
  \centering
  \caption{Decoder at the eavesdropper for the scenario of $p_1\leq p_1^*\leq p_2\leq p_2^*$.}\label{fig:decoder_wiretap_caseii}
\end{figure*}

\textbf{Achievable Rate and Reliability:}

In this way, all information bits and random bits can be recovered reliably, i.e., \eqref{equ:error_bound_main} still holds in this scenario. Meanwhile, we have
\begin{align}
R   &=\frac{1}{NB}\left(|\mathcal{M}_2|\times|\mathcal{A}| +|\mathcal{M}_1|\times|\mathcal{A}^c|\right)\nonumber\\
    &=[H(p_1^*)-H(p_1)]\times q_1 +[H(p_2^*)-H(p_2)]\times q_2,\nonumber
\end{align}
which means the upper bound \eqref{equ:capacity_system} is achieved.

\textbf{Security:}

Assume the receiver from the eavesdropper channel knows all the information bits, i.e., $u^{(i)}_{1:|\mathcal{A}^c|}$ and $v^{(j)}_{1:|\mathcal{A}|}$ in this scenario. Then, the decoder (employed at the eavesdropper) can obtain all random bits by following the steps below (also see Fig.~\ref{fig:decoder_wiretap_caseii}):

\begin{itemize}
\item Phase I (BSC Decoding for the Superior State): Decode the block with respect to the superior state using BSC$(p_1^*)$ SC decoder by knowing all frozen bits in $\mathcal{F}$ and $\mathcal{M}_2$.
\item Phase II (BEC Decoding): Add erasures to the decoded bits in set $\mathcal{M}_1$ and $\mathcal{M}_3$ from previous phase, then decode both the random bits using BEC$(q_2)$ SC decoder by choosing frozen bits as $u^{(i)}_{1:|\mathcal{A}^c|}$ and $0$ respectively.
\item Phase III (BSC Decoding for the Degraded State): Recover all bits in set $\mathcal{M}_1$ and $\mathcal{M}_3$ to make them the frozen bits, and decode the block with respect to the degraded state using BSC$(p_2^*)$ SC decoder.
\end{itemize}

Hence, all random bits can be decoded reliably, i.e., \eqref{equ:error_bound_wiretap} still holds in this scenario. Then, the same procedures as the previous scenario complete the proof of security.

\end{proof}

%%%%%%%%%%%%%%%%%%%%%%%%%%%%%%%%%%%%%%%%%%%%%%%%%%%%%%%%%%%%%%%%%%%%%%%%%%%%%%
\section{Hierarchical Polar Coding for Independent Fading Case}
\label{sec:independent_fading}

In this section, we focus on the case of independent fading for the main channel and the eavesdropper channel. More precisely, the main channel has probability $q_1$ to be in the superior fading state, while the eavesdropper channel has probability $q_1^*$ to be in the superior state (independent of the main channel). The main hurdle here is that the main and eavesdropper channels can be in different fading states (e.g., the main channel can be in degraded state while the eavesdropper channel is in the superior state). Still, as considered in the previous section, we distinguish two scenarios based on the relation between parameters $p^*_1$ and $p_2$.

\subsection{The Scenario of $p_1\leq p_2 \leq p_1^*\leq p_2^*$}

In this scenario, for those fading blocks where the main channel is in degraded state and eavesdropper channel is in superior state, the main channel is still stronger due to $p_2 \leq p_1^*$. To this end, the upper bound for secrecy capacity can be expressed as follows.

\begin{lemma}\label{thm:UpperBound2}
The secrecy capacity for the independent fading scenario with $p_1\leq p_2 \leq p_1^*\leq p_2^*$ is upper bounded by
$$C^s_{\textrm{CSI-D}} \leq q_1^* H(p_1^*)+q_2^* H(p_2^*)-q_1H(p_1)-q_2H(p_2).$$
\end{lemma}

\begin{proof}
We have the following.
\begin{align}
C^s_{\textrm{CSI-D}} & \leq C^s_{\textrm{CSI-ED}}\nonumber\\
                    & =\max_{p(x|s,s^*)} \left[I(\mathsf{X};\mathsf{Y}|\mathsf{S},\mathsf{S}^*)-I(\mathsf{X};\mathsf{Z}|\mathsf{S},\mathsf{S}^*)\right]\nonumber\\
                    & = q_1q_1^* [H(p_1^*)-H(p_1)]+q_1q_2^* [H(p_2^*)-H(p_1)]\nonumber\\
                    &\quad\quad +q_2q_1^* [H(p_1^*)-H(p_2)]+q_2q_2^* [H(p_2^*)-H(p_2)]\nonumber\\
                    & = q_1^* H(p_1^*)+q_2^* H(p_2^*)-q_1H(p_1)-q_2H(p_2),\label{equ:SCbound_case2_scenario1}
\end{align}
where random variables $\mathsf{S}$ and $\mathsf{S}^*$ are the fading states for the main channel and eavesdropper respectively; $q_2=1-q_1$ and $q_2^*=1-q_1^*$.
\end{proof}

The encoder for this independent fading case is similar to the simultaneous fading case (see Fig.~\ref{fig:encoder_case3}), however, the random bits $r^{(i)}_{1:|\mathcal{A}^*|}$ are now of length $|\mathcal{A}^*|$, where set $\mathcal{A}^*$ is the information set for channel BEC$(q^*_2)$, and corresponding decoder at the eavesdropper is SC BEC$(q^*_2)$ decoder.

\begin{figure*}[t]
  \includegraphics[scale=0.78]{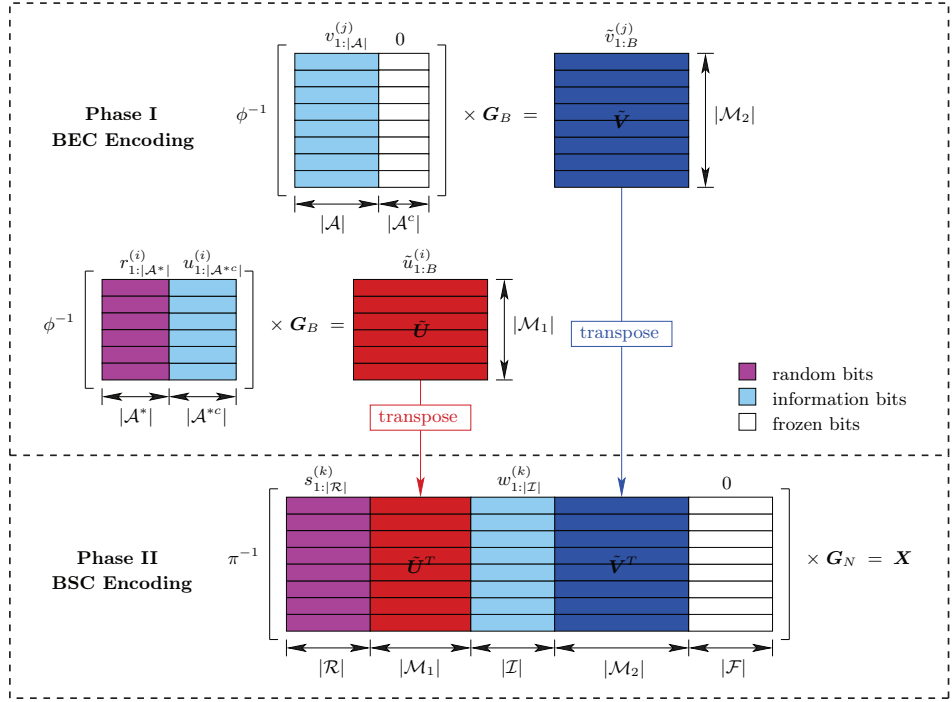}
  \centering
  \caption{Encoder for the scenario of $p_1\leq p_2\leq p_1^*\leq p_2^*$.}\label{fig:encoder_case3}
\end{figure*}

Based on these modifications, all information bits and random bits can still be decoded in this scenario, which implies an achievable rate given by
\begin{align}
R   &=\frac{1}{NB}\left(|\mathcal{M}_2|\times|\mathcal{A}| +|\mathcal{M}_1|\times|\mathcal{A}^{*c}|+|\mathcal{I}|\times B\right)\nonumber\\
    &=[H(p_2)-H(p_1)]\times q_1 +[H(p_2^*)-H(p_1^*)]\times q_2^*+[H(p_1^*)-H(p_2)]\nonumber\\
    &=q_1^* H(p_1^*)+q_2^* H(p_2^*)-q_1H(p_1)-q_2H(p_2).\nonumber
\end{align}
The reliability and security proofs follow from the same steps as the ones detailed for the simultaneous fading case. This achievable rate matches to the upper bound given by Lemma~\ref{thm:UpperBound2}, establishing the secrecy capacity of the system as reported below.

\begin{theorem}
The secrecy capacity for the independent fading scenario with $p_1\leq p_2 \leq p_1^*\leq p_2^*$ is given by
$$C^s_{\textrm{CSI-D}} = q_1^* H(p_1^*)+q_2^* H(p_2^*)-q_1H(p_1)-q_2H(p_2).$$
\end{theorem}

\subsection{The Scenario of $p_1\leq p_1^*\leq p_2\leq p_2^*$}

\textbf{Upper bound on secrecy capacity:}

In this scenario, for those fading blocks where the main channel is in degraded state and eavesdropper channel is in superior state, the eavesdropper channel is stronger. Therefore, the upper bound for secrecy capacity can be expressed as in the following.

\begin{lemma}\label{thm:UpperBound3}
The secrecy capacity for the independent fading scenario with $p_1\leq p_1^*\leq p_2\leq p_2^*$ is upper bounded by
$$C^s_{\textrm{CSI-D}} \leq q_1q_1^* H(p_1^*)+q_2^* H(p_2^*)-q_1H(p_1)-q_2q_2^*H(p_2).$$
\end{lemma}

\begin{figure*}[t]
  \includegraphics[scale=0.7]{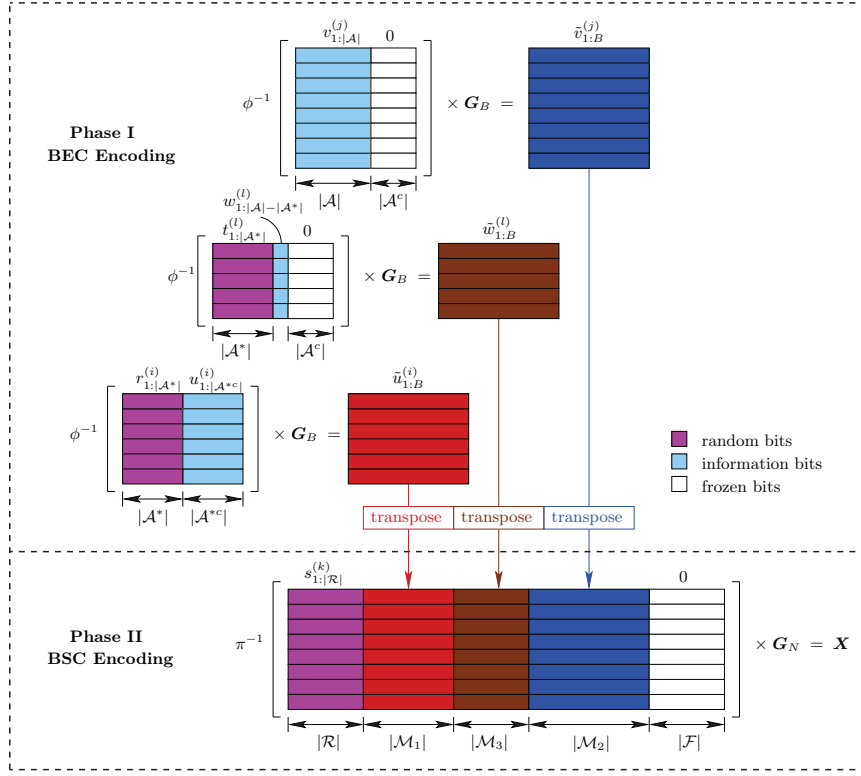}
  \centering
  \caption{Encoder for independent fading case with $q_1\geq q_1^*$}\label{fig:encoder_case4b}
\end{figure*}

\begin{proof}
We have
\begin{align}
C^s_{\textrm{CSI-D}} & \leq C^s_{\textrm{CSI-ED}}\nonumber\\
                    & =\max_{p(x|s,s^*)} \left[I(\mathsf{X};\mathsf{Y}|\mathsf{S},\mathsf{S}^*)-I(\mathsf{X};\mathsf{Z}|\mathsf{S},\mathsf{S}^*)\right]\nonumber\\
                    & = q_1q_1^* [H(p_1^*)-H(p_1)]+q_1q_2^* [H(p_2^*)-H(p_1)]\nonumber\\
                    &\quad\quad +q_2q_1^* 0+q_2q_2^* [H(p_2^*)-H(p_2)]\nonumber\\
                    & = q_1q_1^* H(p_1^*)+q_2^* H(p_2^*)-q_1H(p_1)-q_2q_2^*H(p_2).\label{equ:SCbound_case2_scenario2}
\end{align}
\end{proof}

\textbf{Lower bound on secrecy capacity (for $q_1\geq q_1^*$):}

We focus on the case where superior fading states of the main channel occur more frequently than the superior states of the eavesdropper channel, i.e., $q_1\geq q_1^*$. Under this assumption, we establish an achievable rate as reported below.

\begin{theorem}\label{thm:LowerBound3}
The secrecy capacity for the independent fading scenario with $p_1\leq p_1^*\leq p_2\leq p_2^*$ and $q_1\geq q_1^*$ is lower bounded by
$$C^s_{\textrm{CSI-D}} \geq [H(p_1^*)-H(p_1)]\times q_1 +[H(p_2^*)-H(p_2)]\times q_2^*+[H(p_2)-H(p_1^*)]\times(q_1- q_1^*).$$
\end{theorem}

\begin{proof}

%\begin{figure*}[t]
%  \includegraphics[scale=0.7]{Encoder_Case4a}
%  \centering
%  \caption{Encoder for independent fading case with $q_1< q_1^*$}\label{fig:encoder_case4a}
%\end{figure*}

In the scenario of $q_1\geq q_1^*$, the probability of degraded fading state (the underlying erasure probability in the proposed coding scheme) for the main channel is smaller than that of the eavesdropper channel. Therefore, the polarized indices for these receivers satisfy $\mathcal{A}^*\subseteq\mathcal{A}$. This enables us to construct another set of information bits of size $(|\mathcal{A}|-|\mathcal{A}^*|)\times|\mathcal{M}_3|$ in set $\mathcal{M}_3$ (as shown in Fig.~\ref{fig:encoder_case4b}). Then, the main channel decoder can decode all information bits, and eavesdropper, given the information bits, can also decode random bits. We note that the amount of randomness is equal to the capacity of fading eavesdropper channel, so the analysis detailed in Section~\ref{sec:simultaneous:case1} follows here, establishing the proof of security.
Hence, the achievable secrecy rate is
\begin{align}
R   &=\frac{1}{NB}\left[|\mathcal{M}_2|\times|\mathcal{A}| +|\mathcal{M}_1|\times|\mathcal{A}^{*c}|+|\mathcal{M}_3|\times(|\mathcal{A}|-|\mathcal{A}^*|)\right]\nonumber\\
    &=[H(p_1^*)-H(p_1)]\times q_1 +[H(p_2^*)-H(p_2)]\times q_2^*+[H(p_2)-H(p_1^*)]\times(q_1- q_1^*).\nonumber
\end{align}

\begin{remark}
We note that the technique described above may not be utilized for the case of $q_1<q_1^*$. In particular, the symbols denoted by $\mathcal{M}_3$ in Fig.~\ref{fig:encoder_case4b} has to include full randomness dictated by $\mathcal{A}^*$ to satisfy the secrecy constraint. However, with such a code, the main receiver may not decode these random bits as $\mathcal{A}\subset\mathcal{A}^*.$
\end{remark}

%\textbf{Case ii) $q_1< q_1^*$:} Here, $\mathcal{A}\subset\mathcal{A}^*$, and no further information bits can be constructed between random bits and frozen bits with respect to indices in set $\mathcal{M}_3$ (as shown in Fig.~\ref{fig:encoder_case4a}). Then, the main channel decoder can decode all information bits, and eavesdropper, given the information bits, can also decode random bits. Hence, the achievable rate is
%\begin{align}R   &=\frac{1}{NB}\left[|\mathcal{M}_2|\times|\mathcal{A}| +|\mathcal{M}_1|\times|\mathcal{A}^{*c}|\right]\nonumber\\    &=[H(p_1^*)-H(p_1)]\times q_1 +[H(p_2^*)-H(p_2)]\times q_2^*.\nonumber \end{align}

%Combining the results from both parts, we have the achievable rate for this independent fading case is given by
%\begin{align} R   =[H(p_1^*)-H(p_1)]\times q_1 +[H(p_2^*)-H(p_2)]\times q_2^*+[H(p_2)-H(p_1^*)]\times(q_1- q_1^*)\times \bold{1}_{\{q_1\geq q_1^*\}}.\nonumber \end{align}

\end{proof}

\textbf{On the gap between lower and upper bounds (for $q_1\geq q_1^*$):}

We remark that the rate gap between the upper bound in Lemma~\ref{thm:UpperBound3} and the achievable rate reported in Theorem~\ref{thm:LowerBound3} is given by
\begin{align}
\Delta R
		&= \left\{q_1q_1^* H(p_1^*)+q_2^* H(p_2^*)-q_1H(p_1)-q_2q_2^*H(p_2) \right\}\nonumber\\
		&\quad -\left\{ [H(p_1^*)-H(p_1)]\times q_1 +[H(p_2^*)-H(p_2)]\times q_2^*+[H(p_2)-H(p_1^*)]\times(q_1- q_1^*)\right\}\nonumber\\
        &=q_1^*q_2[H(p_2)-H(p_1^*)]
        %\bold{1}_{\{q_1\geq q_1^*\}}\nonumber\\ &=\left\{\begin{array}{ll} q_1q_2^*[H(p_2)-H(p_1^*)], & \text{if }q_1< q_1^*,\\        q_1^*q_2[H(p_2)-H(p_1^*)], & \text{if }q_1\geq q_1^*.        \end{array}\right.
\end{align}

Noting that we have $q_1\geq q_1^*$ in this scenario, we can further upper bound the gap as follows.
\begin{align}
\Delta R
& =  q_1^*q_2[H(p_2)-H(p_1^*)]\\
& \stackrel{(a)}{\leq} q_1^* (1-q_1^*) [H(p_2)-H(p_1^*)] \\
& \stackrel{(b)}{\leq} 0.25[H(p_2)-H(p_1^*)]\\
& \stackrel{(c)}{\triangleq} \overline{\Delta R}, \label{eq:gapupperbound}
\end{align}
where (a) is due to $q_1\geq q_1^*$, implying $q_2=1-q_1\leq 1-q_1^*$, (b) follows as $\max\limits_{x} x(1-x)=0.25$, and in (c) we define this upper bound on the gap as $\overline{\Delta R}$.

Fig.~\ref{fig:bounds and gap} illustrates the relationships between the upper bound and achievable rate proposed in Lemma~\ref{thm:UpperBound3} and Theorem~\ref{thm:LowerBound3}. In Fig.~\ref{fig:bounds}, we report the \emph{gap coefficient} $\frac{\Delta R}{[H(p_2)-H(p_1^*)]}$ as a function of the fading parameters $q_1$ and $q_1^*$, the probabilities of superior fading states for main and eavesdropper channels, respectively. Note that we have $1\geq q_1\geq q_1^*\geq 0$ in this case. In Fig.~\ref{fig:gap}, we report the upper bound on the gap given by $\overline{\Delta R}$ as a function of channel parameters $p_2$ and $p_1^*$. Note that, we have $0.5\geq p_2\geq p_1^*\geq 0$ in this case. The upper bound on the gap is at most $0.25$ (bits), as can be seen from the expression of $\overline{\Delta R}$ in \eqref{eq:gapupperbound}, increases with $p_2$, and decreases with $p_1^*$. Here, the gap diminishes as $p_1^*\leq p_2$ gets closer to $p_2$. The capacity is established earlier for $p_1^*\geq p_2$ case, i.e., when main channel fading realization is always stronger than that of eavesdropper, so we set the corresponding points for gap to zero in the plot. $\overline{\Delta R}$ is equal to $q_2q_1^*$ times the difference between the channel capacity for superior eavesdropper channel (i.e., $1-H(p_1^*)$) and that for degraded main channel (i.e., $1-H(p_2)$). Thus, this upper bound on gap to capacity for the proposed scheme linearly scales with the difference of these channel capacities.

\begin{figure*}
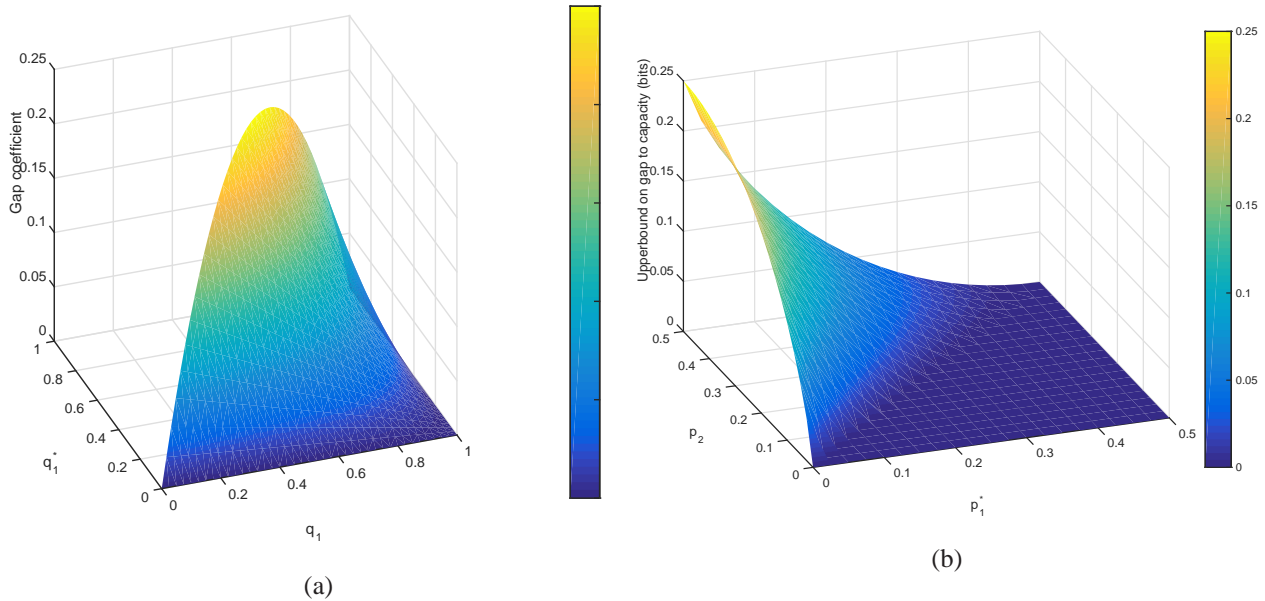

\begin{subfigure}{.5\textwidth}
  \centering
  \includegraphics[width=1\linewidth]{Rate_Gap_Coefficient.eps}
  \caption{}
  \label{fig:bounds}
\end{subfigure}
\begin{subfigure}{.5\textwidth}
  \centering
  \includegraphics[width=1\linewidth]{Rate_Gap_Upperbound.eps}
  \caption{}
  \label{fig:gap}
\end{subfigure}
\caption{Illustration of the gap to capacity. (a) Gap coefficient $\frac{\Delta R}{[H(p_2)-H(p_1^*)]}$ as a function of the fading parameters $q_1$ and $q_1^*$ for $p_1\leq p_1^*\leq p_2\leq p_2^*$. (b) Upper bound $\overline{\Delta R}$ on the gap between achievable rate and capacity. (No gap exists and the capacity is achieved for $p_2\leq p_1^*$.)}
\label{fig:bounds and gap}
\end{figure*}

We remark that the proposed coding scheme pays the penalty of securing information by exhausting the capacity seen by the eavesdropper. That is, the amount of randomness we utilize is equal to the fading channel capacity seen by the eavesdropper (according to the marginal distribution of the channel $p(y^*|x)$). However, this is not always the case for the upper bound. For instance, when the eavesdropper channel realization is superior and the main channel is degraded, as the encoder is assumed to know CSI, no additional penalty is paid to secure the information for these fading blocks as security can not be achieved. So, while the achievability assumes no knowledge of instantaneous CSI, the encoder knows and adapts the code according to eavesdropper CSI for the converse argument. Therefore, the gap we reported here (probably mostly) reflects the loss due to this CSI knowledge difference.

%%%%%%%%%%%%%%%%%%%%%%%%%%%%%%%%%%%%%%%%%%%%%%%%%%%%%%%%%%%%%%%%%%%%%%%%%%%%%%

\section{Conclusion}
\label{sec:conclusion}

In this paper, a hierarchical polar coding scheme is proposed for binary symmetric wiretap channels with block fading. By exploiting an erasure decoding approach at the receiver, this scheme utilizes the polarization of degraded binary symmetric channels to survive from the impact of fading. Meanwhile, to combat with eavesdropping, random bits are injected into the encoded symbols. We showed that this proposed coding scheme achieves the secrecy capacity when both main and eavesdropper channels experience block fading simultaneously. For the scenario of independent block fading model, we showed that the capacity is achieved when the main channel has always a superior fading realization as compared to that of the eavesdropper. For the remaining case of when eavesdropper's state can be stronger than the main receiver, a gap to secrecy capacity is derived using an upper bound derived from a model where the encoder knows the instantaneous CSI and a lower bound for the special case of when superior fading state frequency of the main channel is higher than that of the eavesdropper.

Remarkably, for the cases where the proposed coding scheme achieves the secrecy capacity, there is no loss due to statistical CSI knowledge (as compared to instantaneous CSI knowledge). For the remaining cases, namely when the eavesdropper channel can see stronger channel state than that of the main channel, this conclusion remains an open problem, and not only the inner bound, but also the upper bound we proposed here could be loose. In addition, the case where the superior fading channel frequency of the eavesdropper channel is greater than that of the main channel has resisted our efforts thus far. The hierarchical coding scheme proposed here does not extend to this case (as the required inclusion of polarized channels is not satisfied for this scenario), and this case remains as an open problem. We finally note that, although we consider binary symmetric channels in this paper, the hierarchical coding scheme can be applied as a general method to other scenarios (such as fading blocks with more states) for simultaneously resolving fading and security problems. In particular, noting that AWGN channels with BPSK modulation and demodulation resembles a BSC, the proposed scheme covers a fairly large set of practically relevant channel models.

%================================================================================================================================================

\bibliographystyle{IEEEtran}
\bibliography{References}

\end{document}